\RequirePackage{fix-cm}
\documentclass{svjour3}                     
\smartqed  
\usepackage{graphicx}
\usepackage{amsmath}
\usepackage{amssymb}

\newcommand{\beq}{\begin{equation}} \newcommand{\eeq}{\end{equation}}
\newcommand{\bea}{\begin{eqnarray}} \newcommand{\eea}{\end{eqnarray}}
\newcommand{\bear}{\begin{eqnarray*}} \newcommand{\eear}{\end{eqnarray*}}
\newcommand{\lb}{\label} 
\newcommand{\rf}[1]{(\ref{#1})}

\usepackage{hyperref}
\hypersetup{colorlinks=true, linkcolor=blue, citecolor=blue, filecolor=blue, urlcolor=blue,
            pdfproducer={Latex},
            pdfcreator={pdflatex}}

\begin{document}

\title {Noether theorem for action-dependent Lagrangian functions: conservation laws for non-conservative systems}

%
%

\author{M.~J.~Lazo \and J.~Paiva \and G.~S.~F.~Frederico
}


\institute{M.~J.~Lazo,  Corresponding author \at
              Instituto de Matem\'atica, Estat\'istica e F\'isica, Universidade Federal do Rio Grande, Rio Grande - RS, Brazil.\\
              \email{matheuslazo@furg.br}
              \and
              J.~Paiva \at
              Instituto de Matem\'atica, Estat\'istica e F\'isica, Universidade Federal do Rio Grande, Rio Grande - RS, Brazil.\\
              \email{juilsonpaiva@hotmail.com}
              \and
              G.~S.~F.~Frederico \at
              Universidade Federal do Cear\'a, Campus de Russas, Russas, Brazil. \\
              \email{gastao.frederico@ua.pt}             \\
}

\date{Received: date / Accepted: date}

\maketitle

\begin{abstract}

In the present work, we formulate a generalization of the Noether Theorem for action-dependent Lagrangian functions. The Noether's theorem is one of the most important theorems for physics. It is well known that all conservation laws, \textrm{e.g.}, conservation of energy and momentum, are directly related to the invariance of the action under a family of transformations. However, the classical Noether theorem cannot be applied to study non-conservative systems because it is not possible to formulate physically meaningful Lagrangian functions for this kind of systems in the classical calculus of variation. On the other hand, recently it was shown that an Action Principle with action-dependent Lagrangian functions provides physically meaningful Lagrangian functions for a huge variety of non-conservative systems (classical and quantum). Consequently, the generalized Noether Theorem we present enable us to investigate conservation laws of non-conservative systems. In order to illustrate the potential of application, we consider three examples of dissipative systems and we analyze the conservation laws related to spacetime transformations and internal symmetries.



\end{abstract}

\maketitle


\section{Introduction}

Since the introduction of the Action Principle in its mature formulation by Euler, Lagrange, and Hamilton, it has become one of the most fundamental principles of physics. It provides a solid and universal foundation for the whole dynamical structure in any classical or quantum conservative theory. Actually, it is from the Action Principle that the dynamical equations describing any conservative system, in any physical theory (classical or quantum), is obtained. However, despite its importance to study conservative systems, it is well known that the equations of motion for dissipative linear dynamical systems cannot be obtained from a physically meaningful Lagrangian in the classical Action Principle framework. A rigorous proof for the failure of the Action Principle in describing dissipative systems was given in 1931 by Bauer \cite{bauer}, who proved that it is impossible to obtain a dissipation term proportional to the first order time derivative in the equation of motion from the traditional Action Principle. 

Over the last century, several methods have been developed in order to overcome the failure of the Action Principle to describe non-conservative systems. Examples include time-dependent Lagrangians \cite{Ktevens,Havas,Negro,Negro2,Brinati,Tartaglia}, the Bateman approach by introducing auxiliary coordinates that describe the reverse-time system \cite{Bateman,MF,FT,CRTV,CRV,VujaJones} and Actions with fractional derivatives \cite{Riewe,LazoCesar}. Unfortunately, all these approaches either give us non-physically meaningful Lagrangian functions (in the sense that they provide non-physical relations for the momentum and Hamiltonian of the system) or make use of non-local differential operators with algebraic properties different from usual derivatives (see \cite{Riewe,LazoCesar} for a detailed discussion). In recent works \cite{MJJG,MJJG2}, in order to formulate an Action Principle for non-conservative systems, we take a different approach and we propose a physically meaningful Action Principle for dissipative systems by generalizing the variational problem \cite{Herg1,Herg2,GGB} for  Lagrangian density functions depending itself on an action-density field. In any physical theory, the Lagrangian function which defines the Action is constructed from the scalars of the theory, and from it, the corresponding dynamical equations can be obtained. However, the Action itself is a scalar and we might ask ourselves what would happen if the Lagrangian function itself were a function of the Action. For a one dimensional system, the answer to this question can be given by an almost forgotten variational problem proposed by Herglotz in 1930 \cite{Herg1,Herg2,GGB}. A reason for this problem to be almost unknown is that a covariant generalization for several independent variables is not direct. Only
recently the Herglotz variational problem gained more interest in the literature \cite{Herg2,GGB,GGB2,SMT,Zhang1,Zhang2,Zhang3,Zhang4,Zhang5} and, in particular,
 in a recent work \cite{MJJG} we formulated a covariant generalization for the Herglotz problem to construct a non-conservative gravitational theory from the Lagrangian formalism. Furthermore, by following the ideas we introduced in \cite{MJJG}, in \cite{MJJG2} we formulated a general Action Principle for non-conservative systems for Lagrangian density functions depending itself on an action-density field. We obtained a generalization of the Euler-Lagrange equation for this Action Principle and applied it in several classical and quantum systems. 

A physically consistent generalization of the Action Principle to non-conservative systems enable us to use all the mathematical machinery of the calculus of variation to study dissipative systems. Among these mathematical machineries, we can highlight the Noether's theorem that becomes one of the most important theorems for physics in the 20th century. Since the seminal work of Emmy Noether, it is well known that all conservation's laws in mechanics, as for example, conservation of energy and momentum, are directly related to the invariance of the Action under a family of transformations. Furthermore, conserved quantities in any dynamical system play a major role in the analysis of the system. They enable us to solve some problems without a more detailed knowledge of the dynamics, for example, as found in any undergraduate textbook in physics, we can solve easily several mechanical problems by making use of the energy and momentum conservations without the necessity of solving the dynamical equation given by Newton's second law of motion. In general, the continuous symmetries and its related conserved quantities give us first integrals for dynamical systems that can be used to simplify the problem. 
On the other hand, non-conservative forces remove energy from the systems and, as a consequence, the standard Noether constants of motion (as energy and momentum) are broken. In this context, the generalization of the Noether's theorem for the Action Principle defined in \cite{MJJG,MJJG2} is fundamental to investigate the action symmetries for nonconservative systems. In the present work, we generalize Noether's theorem for Lagrangian density functions depending itself on an action-density field. As examples of application to non-conservative systems, we study the problem of a vibrating string under a frictional force, and the problem of a complex scalar field.

The paper is organized in the following way. In Section \ref{sec:AP} we review the basic notions of Herglotz variational problem and our generalization for fields \cite{MJJG,MJJG2}. Furthermore, we discuss the gauge invariance in our Action Principle \cite{MJJG,MJJG2} and we introduce a Canonical Gauge. The Noether's theorem for Lagrangian density functions depending itself on an action-density field is obtained in Section \ref{sec:NT}. The examples of applications of the Noether's theorem are presented in Section \ref{sec:EX}. Finally, the conclusions are presented in Section \ref{sec:C}.

\section{Action Principle for Action dependent Lagrangians}\lb{sec:AP}

In this section, we first present the Action Principle introduced by us in \cite{MJJG,MJJG2}, and after we show that the Action is gauge invariant. This gauge invariance will play a fundamental role in the generalization of the Noether Theorem.

\subsection{The Herglotz variational problem}\lb{sec:APb}

In recent works \cite{MJJG,MJJG2}, we formulated a covariant Action Principle for Action dependent Lagrangian densities by generalizing the Herglotz variational problem. The original Herglotz problem, introduced in 1930 \cite{Herg1,Herg2}, consists in the problem of determining the function $x(t)$ that  extremizes (minimizes or maximizes) $S(b)$, where the Action $S(t)$ is a solution of
\beq
\lb{H}
\begin{split}
\dot{S}(t)&=L(t,x(t),\dot{x}(t),S(t)),\;\;\; t\in [a,b]\\
S(a)&=s_a, \;\;\; x(a)=x_a, \;\;\; x(b)=x_b, \;\;\; s_a,x_a,x_b \in \mathbb{R}.
\end{split}
\eeq
It is important to stress that $S(t)$ is a functional since, for each function $x(t)$, we have a different differential equation problem \eqref{H}. Therefore, $S(t)$ depends on $x(t)$. Furthermore, the Herglotz problem \eqref{H} reduces to the classical fundamental problem of the calculus of variations if the Lagrangian function $L$ does not depend on $S(t)$. In this particular case, by integrating \eqref{H}, we obtain the classical variational problem
\beq
\lb{H2}
S(b)=\int_a^b \tilde{L}(t,x(t),\dot{x}(t))\;dt\longrightarrow {\mbox{extremum}},
\eeq
where $x(a)=x_a$, $x(b)=x_b$, and
\beq
\lb{H3}
\tilde{L}(t,x(t),\dot{x}(t))=L(t,x(t),\dot{x}(t)) +\frac{s_a}{b-a}.
\eeq

Herglotz proved \cite{Herg1,Herg2} that a necessary condition for a function $x(t)$ to imply an extremum of the variational problem \eqref{H} is given by the generalized Euler-Lagrange equation:
\beq
\lb{HEL}
\frac{\partial L}{\partial x} -\frac{d}{dt}\frac{\partial L}{\partial \dot{x}}+\frac{\partial L}{\partial S}\frac{\partial L}{\partial \dot{x}}=0.
\eeq
It is easy to notice that in the case where $\frac{\partial L}{\partial S}=0$, as in the classical problem \eqref{H2}, the differential equation \eqref{HEL} reduces to the classical Euler-Lagrange equation. The potential application of Herglotz problem to non-conservative systems is evident even in the simplest case, where the dependence of the Lagrangian function on the Action is linear \cite{MJJG2}. For example, the Lagrangian function
\beq
\lb{H3b}
L=\frac{m\dot{x}^ 2}{2}-U(x)-\frac{\gamma}{m} S
\eeq
describes a dissipative system with a point particle of mass $m$ under a potential $U(x)$ and a viscous force with a resistance coefficient $\gamma$. From \eqref{HEL}, the resulting equation of motion,
\beq
\lb{H3c}
m\ddot{x}+\gamma \dot{x}=F,
\eeq
includes the well-known dissipative term proportional to the velocity $\dot{x}$, where $\ddot{x}$ is the particle acceleration and $F=-\frac{d U}{dx}$ is the external force. In this context, the linear term $\frac{\gamma}{m} S$ in the Lagrangian function \eqref{H3b} can be interpreted as a potential function for the non-conservative force \cite{MJJG2}. Furthermore, the Lagrangian given by \eqref{H3b} is physical in the sense it provides us with physically meaningful relations for the momentum and the Hamiltonian \cite{MJJG2,Riewe,LazoCesar}. 

\subsection{Generalization of the Herglotz problem for fields}

Although the Herglotz problem was introduced in 1930, a covariant generalization of \eqref{H} for several independent variables is not direct and was proposed only recently \cite{MJJG,MJJG2}. For a scalar field $\phi(x^{\mu})=\phi(x^1,x^2,\cdots,x^d)$ defined in a domain $\Omega \in \mathbb{R}^d$ ($d=1,2,3,\cdots$), the classical problem of variational calculus deals with the problem to find $\phi$ that extremizes the functional
\beq
\lb{H4}
S(\delta\Omega)=\int_{\delta\Omega}\mathcal{L}\left(x^\mu,\phi(x^\mu),\partial_\nu\phi(x^\mu)\right)d^dx,
\eeq
where $\delta\Omega$ is the boundary of $\Omega$, and $\phi$ satisfies the boundary condition $\phi(\delta\Omega)=\phi_{\delta\Omega}$ with $\phi_{\Omega}:\delta \Omega \longrightarrow \mathbb{R}$. The cornerstone of a generalization of the Herglotz problem for fields is to note that, for a given fixed $\phi$, the functional $S$ defined in \eqref{H4} reduces to a function of the boundary $\delta \Omega$. In this context, if there is a differentiable vector field $s^{\mu}$ such that
\beq
\lb{H5}
S(\delta\Omega)=\int_{\delta\Omega}s^\nu n_\nu\; d\sigma,
\eeq
where here and throughout the rest of the work we assume the summation convention on repeated indices, then, from the Divergence Theorem we obtain
\beq
\lb{H6}
S(\delta\Omega)=\int_{\delta\Omega}s^\nu n_\nu d\sigma=\int_{\Omega}\partial_\nu s^\nu d^dx=\int_{\Omega}\mathcal{L}\left(x^\mu,\phi(x^\mu),\partial_\nu\phi(x^\mu),s^\mu\right) d^dx,
\eeq
where we consider that $\delta\Omega$ is an orientable Jordan surface, $n_{\mu}$ is the surface normal vector field, and $d\sigma$ is the surface differential. Consequently, we can generalize the Herglotz variational principle as follows \cite{MJJG2}:

\begin{definition}[Fundamental Problem]
\lb{HF}
Let the action-density field $s^{\mu}$ be a differentiable vector field on $\Omega \in \mathbb{R}^d$. The fundamental problem of Herglotz variational principle for fields consists in determining the field $\phi$ that  extremizes (minimizes or maximizes) $S(\delta\Omega)$, where $S(\delta\Omega)$ is given by
\beq
\lb{H7}
\begin{split}
&\partial_\nu s^\nu=\mathcal{L}\left(x^\mu,\phi(x^\mu),\partial_\mu\phi(x^\mu),s^\mu\right),\quad x^\mu=(x^1, x^2,..., x^d)\in\Omega\\
&S(\delta\Omega)=\int_{\delta\Omega}s^\nu n_\nu d\sigma,\quad\phi(\delta\Omega)=\phi_{\delta\Omega},\quad \phi(\delta\Omega):\delta\Omega\longrightarrow \mathbb{R}.
\end{split}
\eeq
\end{definition}

Like in the original Herglotz problem, it is easy to notice that the Action functional defined by \eqref{H7} reduces to the usual Action \eqref{H4} when the Lagrangian function is independent of the action-density field $s^{\mu}$. Furthermore, we can prove the following condition for the extremum of \eqref{H7} (see \cite{MJJG2} for the proof):

\begin{theorem}[Generalized Euler-Lagrange equation for non-conservative fields]
\lb{HELF}
Let $\partial_{s^\nu}\mathcal{L}=\gamma_\nu$ be a gradient $\gamma_\nu=\partial_{\nu}f(x^\mu)=(\partial_{x_1}f,\cdots,\partial_{x_d}f)$ of a scalar field $f:\Omega \longrightarrow \mathbb{R}$, and let $\phi^*$ be the fields that extremize $S(\delta\Omega)$ defined in \eqref{H7}. Then, the field $\phi^*$ satisfies the generalized Euler-Lagrange equation
\beq
\label{GHEL}
\dfrac{\partial\mathcal{L}}{\partial \phi^* }-\frac{d}{dx^\nu}\left(\dfrac{\partial \mathcal{L} }{\partial\left(\partial_\nu\phi^*\right)}\right)+\gamma_\nu \dfrac{\partial \mathcal{L} }{\partial\left(\partial_\nu\phi^*\right)}=0.
\eeq
\end{theorem}

It is easy to see that for Lagrangian functions independent on $s^{\mu}$, the generalized Euler-Lagrange equation \rf{GHEL} reduces to the usual one,
\beq
\lb{HEL2}
\dfrac{\partial\mathcal{L}}{\partial \phi^* }-\frac{d}{dx^\nu}\left(\dfrac{\partial \mathcal{L} }{\partial\left(\partial_\nu\phi^*\right)}\right)=0,
\eeq
since, in this case, $\gamma_\mu =0$. Furthermore, when the action-density field $s^{\mu}$ has only one non-null component and it is a function of only one variable, for example $s^1\neq 0$ and $x^1=t$, and $\Omega=[t_a,t_b]\otimes \mathbb{R}^{d-1}$, the fundamental problem in Definition \ref{HF} contains, as a particular case, the non-covariant problem introduced in \cite{GGB}. Moreover, in the latter situation, equation \eqref{H7} can be easily solved for Lagrangian functions linear on $s^1$, resulting in a $s^1$ expressed as a history-dependent function.

Finally, it is straightforward to generalize  the fundamental problem \eqref{H7} and the Euler-Lagrange equation \eqref{GHEL} to the case with several fields $\phi^i(x^{\mu})=\phi^i(x^1,x^2,\cdots,x^d)$ ($i=1,...,N$). In this case we have the Action $S(\delta\Omega)$ defined by \cite{MJJG2}
\beq
\lb{H7B}
\begin{split}
&\partial_\nu s^\nu=\mathcal{L}\left(x^\mu,\phi^i(x^\mu),\partial_\mu\phi^i(x^\mu),s^\mu\right),\quad x^\mu=(x^1, x^2,..., x^d)\in\Omega\\
&S(\delta\Omega)=\int_{\delta\Omega}s^\nu n_\nu d\sigma,\quad\phi^i(\delta\Omega)=\phi^i_{\delta\Omega},\quad \phi^i(\delta\Omega):\delta\Omega\longrightarrow \mathbb{R},
\end{split}
\eeq
and for a Lagrangian function for which $\partial_{s^\nu}\mathcal{L}=\gamma_\nu=\partial_{\nu}f(x^\mu)$, we obtain the following set of generalized Euler-Lagrange equations:
\beq
\label{GHEL2}
\dfrac{\partial\mathcal{L}}{\partial \phi^{i*} }-\frac{d}{dx^\nu}\left(\dfrac{\partial \mathcal{L} }{\partial\left(\partial_\nu\phi^{i*}\right)}\right)+\gamma_\nu \dfrac{\partial \mathcal{L} }{\partial\left(\partial_\nu\phi^{i*}\right)}=0,\quad i=1,...,N.
\eeq

We can now formulate an Action Principle suited to dissipative systems and free from difficulties found in previous approaches.
\begin{definition}[Generalized Action Principle \cite{MJJG2}]
\lb{AP}
The equation of motion for a physical field $\phi^i$ is the one for which the Action \eqref{H7B} is stationary.
\end{definition}

As a consequence of Definition \ref{AP}, the physical field should satisfy the generalized Euler-Lagrange equation \eqref{GHEL2}. Since for Lagrangian functions independent on the action-density the variational problem \eqref{H7B} reduces to the classical one, the generalized Action Principle is appropriate to describe both conservative and non-conservative systems \cite{MJJG2}.

\begin{remark} We can extend for the more general Action Principle in Definition \ref{AP} the physical interpretation we give to the Lagrangian \eqref{H3b} for a single particle under frictional forces. Since in Theorem \ref{HELF} we consider only the particular case when $\partial_{s^\nu}\mathcal{L}=\gamma_\nu$ with $\gamma_\nu=\partial_{\nu}f(x^\mu)=(\partial_{x_1}f,\cdots,\partial_{x_d}f)$, the Lagrangian functions $\mathcal{L}$ of a general physical system, in the sense of Definition \ref{AP}, can be written as $\mathcal{L}=\mathcal{L}_c+\gamma_\nu s^\nu$, where $\mathcal{L}_c$ is a standard Lagrangian function for the corresponding conservative system (kinetic energy minus the potential of conservative interactions) and $\gamma_\nu s^\nu$ can be interpreted as the potential energy of non-conservative interactions. In this context, the physical content of the dissipative interactions are contained in $\gamma_\nu$, that is, in the function $f(x^\mu)$. As we shall see in the examples displayed in Section \ref{sec:EX}, for a constant $\gamma_\nu$, and a Lagrangian quadratic in the first order time derivative, we will have a linear frictional force proportional to the first order time derivative,  as in the Lagrangian \eqref{H3b}. This results follows from the fact that the Euler-Lagrange equation give us a frictional force $\gamma_\nu \dfrac{\partial \mathcal{L}_c }{\partial\left(\partial_\nu\phi^{i*}\right)}$ in this case. More general nonlinear dissipative forces will be found when $\gamma_\nu$ is not constant and when the Lagrangian $\mathcal{L}_c$ is not quadratic in the first order time derivative.
\end{remark}

\subsection{Gauge invariance of the Action}

An important and interesting feature of the Action Principle in Definition \ref{AP} is that the Action $S(\delta \Omega)$ is gauge invariant. This follows directly from the fact that the Fundamental Problem \eqref{H7} do not completely fix the action-density field $s^\mu$. For example, in the three-dimensional case where $x^\mu=\vec{x}=(x^1,x^2,x^3)$ and $s^\mu=\vec{s}=(s^1,s^2,s^3)$, the action-density field $s^\mu$ is fixed unless the curl of any vector field $\vec{v}$, since for any $\vec{\overline{s}}=\vec{s}+\nabla\times \vec{v}$ we have
\beq
\lb{GI}
\overline{S}(\delta\Omega)=\int_{\delta\Omega}\vec{\overline{s}}\cdot \vec{n}\; d\sigma=\int_{\Omega}\nabla\cdot \vec{\overline{s}}\; d^3x=\int_{\Omega}\nabla\cdot \left( \vec{s}+\nabla\times \vec{v}\right) d^3x=\int_{\Omega}\nabla\cdot \vec{s}\; d^3x=S(\delta\Omega),
\eeq
where $\overline{S}(\delta\Omega)$ is defined by \eqref{H7} with $\vec{\overline{s}}$ instead of $\vec{s}$. Although the action-density field $s^\mu$ is not completely fixed by \eqref{H7}, the Fundamental Problem \eqref{H7} give us dynamic equations \eqref{GHEL} that completely fix the field $\phi$. We can make an analogy with the electromagnetic theory, where the Maxwell equations do not completely fix the four-vector potential, and say that the Action $S(\delta \Omega)$ is gauge invariant under gauge transformations of the action-density field $s^\mu$. Let us define the following Canonical Gauge:
\begin{definition}[Canonical Gauge]
Let $s^\mu$ be a differentiable vector field and $\gamma_\mu=\partial_{s^\mu}\mathcal{L}$. The Canonical Gauge for the Fundamental Problem \eqref{H7} is defined by the condition
\beq
\lb{CG}
\gamma_\nu\partial_\mu s^\mu-\gamma_\mu\partial_\nu s^\mu=0.
\eeq
\end{definition}
The reason why we choose the Canonical Gauge as in \eqref{CG} will be clear in the next section when we obtain the generalization of the Noether Theorem. Although the field $\phi$ is independent the gauge we choice, the gauge invariance will play a fundamental role in the generalization of the Noether Theorem.


\section{Generalized Noether Theorem}\lb{sec:NT}

Physical systems described by the Herglotz Euler-Lagrange equation are, in general, non-conservative in the classical sense (as an example, the total energy is non-conserved in systems under frictional forces). In this context, the generalization of Noether Theorem is fundamental to study conservative quantities in non-conservative systems described by Herglotz problems. In recent works, Noether's like theorems for several kinds of Herglotz variational problems are proposed \cite{GGB,GGB2,SMT,Zhang1,Zhang2,Zhang3,Zhang4,Zhang5}. In the present work, in order to generalize the Noether Theorem for our Action Principle given by the Fundamental Problem \ref{HF}, we consider invariance transformations in the $(x^ \mu,\phi)$-space, depending on a real parameter $\epsilon$. To be more precise, we consider the one-parameter group of invertible transformations
\beq
\label{N01}
\left\{
\begin{array}{l}
\tilde{x}^\mu=\varphi^\mu(x^\nu,\phi;\epsilon), \;\;\;\;\; \mu,\nu=1,...,d\\
\tilde{\phi}=\psi(x^\nu,\phi;\epsilon),
\end{array}
\right.
\eeq
where $\varphi^\mu(x^\nu,\phi;0)=x^\mu$ and $\psi(x^\nu,\phi;0)=\phi$.
We now define the transformed action-density $\tilde{s}^\mu$ of $s^\mu$, given by Definition \ref{HF}, as
\begin{definition}
The transformed action-density $\tilde{s}^\mu$ of $s^\mu$, given by Definition \ref{HF}, is a solution of the transformed differential equation
\beq
\lb{N02}
\tilde{\partial}_\nu \tilde{s}^\nu=\mathcal{L}\left(\tilde{x}^\mu,\tilde{\phi},\tilde{\partial}_\mu\tilde{\phi},\tilde{s}^\mu\right),\;\;\; (\tilde{x}^\mu\in \tilde{\Omega}),
\eeq
where $\tilde{\Omega}$ is the transformed domain of the domain $\Omega$, and $\tilde{\partial_\mu}=\frac{\partial}{\partial \tilde{x}^\mu}$.
\end{definition}

We can now define what means a functional $S(\delta D)$ be invariant under a one-parameter group of invertible transformation \rf{N01}.
\begin{definition}[Invariance] \lb{ID}
Let $D$ be a closed subdomain of $\Omega$ with boundary $\delta D\subset \Omega$. We say that the functional defined by \rf{H7} is invariant under the family of transformations \rf{N01} if the functional 
$\tilde{S}(\delta \tilde{D})=\int_{\delta \tilde{D}}\tilde{s}^\nu \tilde{n}_\nu\; d\tilde{\sigma}$
defined by the transformed equation \rf{N02}, where $\tilde{D}$ and $\delta \tilde{D}$ are the transformed $D$ and $\delta D$ under \rf{N01}, and the functional
$S(\delta D)$, defined by the non-transformed equation \rf{H7}, satisfy
\beq
\lb{N04}
\tilde{S}(\delta \tilde{D})=\int_{\delta \tilde{D}}\tilde{s}^\nu\tilde{n}_\nu\; d\tilde{\sigma}=\int_{\delta D}\tilde{s}^\nu n_\nu\; d\sigma=\int_{\delta D}s^\nu n_\nu\; d\sigma=S(\delta D).
\eeq
Furthermore, if the functional is invariant under a local one-parameter group $G$ of transformations \rf{N01}, we say that $G$ is a symmetry group.
\end{definition}

We can now obtain the following identity from the invariance Definition \ref{ID}:
\begin{theorem}\lb{TNHG} Let \rf{N01} be a symmetry group of the functional $S(\delta \Omega)$ defined in \rf{H7}. Then the following identity
\beq
\lb{GNT}
\frac{d}{dx^\nu} \left[\left(\frac{\partial \mathcal{L}}{\partial \left( \partial_\nu \phi\right)}\eta + \mathcal{L}\xi^\nu -\frac{\partial \mathcal{L}}{\partial \left( \partial_\nu \phi\right)}\partial_\mu \phi\; \xi^\mu \right)e^{-f(x^\alpha)}\right]+e^{-f(x^\alpha)}\left(\gamma_\nu\partial_\mu s^\mu-\gamma_\mu\partial_\nu s^\mu \right)\xi^\nu=0
\eeq
holds on solutions of the generalized Euler-Lagrange equation \rf{GHEL}, where $\xi^\mu=\frac{d\varphi^\mu}{d\epsilon}\vert_{\epsilon=0}$ and $\eta=\frac{d\psi}{d\epsilon}\vert_{\epsilon=0}$.
\end{theorem}

\begin{proof}
In order to prove the generalized Noether’s theorem for the Fundamental Problem \rf{H7}, we recall from the Lie theory that near the identity transformation the action of the group \rf{N01} is the same as the action of the infinitesimal linear group
\beq
\lb{N05}
\left\{
\begin{array}{l}
\tilde{x}^\mu=x^\mu+\xi^\mu(x^\nu,\phi)\epsilon, \\
\tilde{\phi}=\phi+\eta(x^\nu,\phi)\epsilon,
\end{array}
\right.
\eeq
where $\xi^\mu=\frac{d\varphi^\mu}{d\epsilon}\vert_{\epsilon=0}$ and $\eta=\frac{d\psi}{d\epsilon}\vert_{\epsilon=0}$. From the condition \rf{N04} we have
\beq
\lb{N06}
\tilde{S}(\delta \tilde{D})=\int_{\delta \tilde{D}}\tilde{s}^\nu\tilde{n}_\nu\; d\tilde{\sigma}=\int_{\tilde{D}}\tilde{\partial}_\nu \tilde{s}^\nu\;d^d\tilde{x}=\int_{\tilde{D}}\mathcal{L}(\tilde{x}^\mu,\tilde{\phi},\tilde{\partial}_\mu\tilde{\phi},\tilde{s}^\mu)\;d^d\tilde{x}=\int_{\delta D}\tilde{s}^ \nu n_\nu\; d\sigma=\int_{D}\partial_\nu \tilde{s}^\nu\; d^dx.
\eeq
Now we perform a change of variables in \rf{N06} to go back to the original variables $x^\mu$ and $\phi$. We obtain
\beq
\lb{N07}
\int_{D}\partial_\nu \tilde{s}^\nu\; d^dx=\int_{D}\mathcal{L}(\tilde{x}^\mu,\tilde{\phi},\tilde{\partial}_\mu\tilde{\phi},\tilde{s}^\mu)\det\left(\frac{\partial \tilde{x}^\mu}{\partial x^\nu}\right)\;d^dx,
\eeq
where the determinant of the Jacobi matrix arises from the change of variables $x^\mu$. By taking a derivative of \rf{N07} with respect to $\epsilon$ we get
\beq
\lb{N08}
\left.\frac{d \tilde{S}(\delta \tilde{D})}{d\epsilon}\right|_{\epsilon=0}=\int_{D}\left.\partial_\nu  \frac{d\tilde{s}^\nu}{d\epsilon}\right|_{\epsilon=0} d^dx=\int_{D}\left.\left[\frac{d\mathcal{L}}{d \epsilon}\det\left(\frac{\partial \tilde{x}^\mu}{\partial x^\nu}\right)+\mathcal{L}\frac{d}{d \epsilon}\det\left(\frac{\partial \tilde{x}^\mu}{\partial x^\nu}\right)\right]\right|_{\epsilon=0}d^dx=0,
\eeq
since by hypothesis the one-parameter group of transformations \rf{N05} leaves the functional $S(\delta D)$ invariant, namely $\tilde{S}(\delta \tilde{D})=S(\delta D)$. Since from \rf{N05} we have $\left.\frac{\partial \tilde{x}^\mu}{\partial x^\nu}\right|_{\epsilon=0}=\delta_{ij}$, we get
\beq
\lb{N09}
\left.\det\left(\frac{\partial \tilde{x}^\mu}{\partial x^\nu}\right)\right|_{\epsilon=0}=1.
\eeq
After some calculations we also obtain
\beq
\lb{N10}
\frac{d}{d\epsilon}\left.\det\left(\frac{\partial \tilde{x}^\mu}{\partial x^\nu}\right)\right|_{\epsilon=0}=\frac{d \xi^\nu}{dx^\nu}.
\eeq
By inserting \rf{N09} and \rf{N10} into \rf{N08} we have
\beq
\lb{N11}
\int_{D}\left.\partial_\nu  \frac{d\tilde{s}^\nu}{d\epsilon}\right|_{\epsilon=0} d^dx=\int_{D}\left[\left.\frac{d\mathcal{L}}{d \epsilon}\right|_{\epsilon=0}+\mathcal{L}\frac{d\xi^\nu}{dx^\nu}\right]d^dx.
\eeq
Thus
\beq
\lb{N12}
\int_{D}\partial_\nu \zeta^\nu \; d^dx=\int_{D}\left[\frac{\partial \mathcal{L}}{\partial x^\nu}\xi^\nu+\frac{\partial \mathcal{L}}{\partial \phi}\eta+\frac{\partial \mathcal{L}}{\partial \left( \partial_\nu \phi\right)}\frac{d}{d\epsilon}\left.\left(\tilde{\partial}_\nu \tilde{\phi}\right)\right|_{\epsilon=0}+\gamma_\nu \zeta^\nu+\mathcal{L}\frac{d\xi^\nu}{dx^\nu}\right]d^dx,
\eeq
where now $\zeta^\mu=\left.\frac{d \tilde{s}^\mu}{d \epsilon}\right|_{\epsilon=0}$. After some calculations we get
\beq
\lb{N13}
\frac{d}{d\epsilon}\left.\left(\tilde{\partial}_\nu \tilde{\phi}\right)\right|_{\epsilon=0}=\frac{d\eta}{dx^\nu}-\partial_\mu \phi\frac{d\xi^\mu}{dx^\nu}.
\eeq
Finally, by inserting \rf{N13} into \rf{N12} results in
\beq
\lb{N14}
\int_{D}\left[\partial_\nu \zeta^\nu-\frac{\partial \mathcal{L}}{\partial x^\nu}\xi^\nu-\frac{\partial \mathcal{L}}{\partial \phi}\eta-\frac{\partial \mathcal{L}}{\partial \left( \partial_\nu \phi\right)}\left(\frac{d\eta}{dx^\nu}-\partial_\mu \phi\frac{d\xi^\mu}{dx^\nu}\right)-\mathcal{L}\frac{d\xi^\nu}{dx^\nu}-\gamma_\nu \zeta^\nu\right]d^dx=0.
\eeq
A sufficient condition to satisfy \rf{N14} for any subdomain $D\subset \Omega$ is that
\beq
\lb{N15}
\partial_\nu \zeta^\nu=\frac{\partial \mathcal{L}}{\partial x^\nu}\xi^\nu+\frac{\partial \mathcal{L}}{\partial \phi}\eta+\frac{\partial \mathcal{L}}{\partial \left( \partial_\nu \phi\right)}\left(\frac{d\eta}{dx^\nu}-\partial_\mu \phi\frac{d\xi^\mu}{dx^\nu}\right)+\mathcal{L}\frac{d\xi^\nu}{dx^\nu}+\gamma_\nu \zeta^\nu 
\eeq
Since $\gamma_\nu=\partial_\nu f(x^\mu)$ is a gradient vector on $D\subset \Omega$, \eqref{N15} implies that $\zeta^\nu$ can be written as
\beq
\lb{N16}
\zeta^\nu(x^\mu,\phi,\partial_\mu\phi,s^\mu)= A^\nu(x^\mu,\phi,\partial_\mu\phi,s^\mu)e^{f(x^\alpha)},
\eeq
where
\beq
\lb{N17}\partial_\nu A^\nu(x^\mu,\phi,\partial_\mu\phi,s^\mu)=\left[\frac{\partial \mathcal{L}}{\partial x^\nu}\xi^\nu+\frac{\partial \mathcal{L}}{\partial \phi}\eta+\frac{\partial \mathcal{L}}{\partial \left( \partial_\nu \phi\right)}\left(\frac{d\eta}{dx^\nu}-\partial_\mu \phi\frac{d\xi^\mu}{dx^\nu}\right)+\mathcal{L}\frac{d\xi^\nu}{dx^\nu}\right]e^{-f(x^\alpha)}.
\eeq
From \rf{N06} we have
\beq
\lb{N18}
\left.\frac{d\tilde{S}(\delta \tilde{D})}{d\epsilon}\right|_{\epsilon=0}=\int_{\delta D}\left.\frac{d\tilde{s}^\nu}{d\epsilon}\right|_{\epsilon=0} n_\nu\; d\sigma=
\int_{\delta D} \zeta^\nu n_\nu\; d\sigma=0.
\eeq
A sufficient condition to satisfy \rf{N18} for any subdomain $D\subset \Omega$ is that the function $A^\nu(x^\mu,\phi,\partial_\mu\phi,s^\mu)$ satisfy the boundary condition $A^\nu n_\nu=0$ for all $x^\mu \in \delta D$. Consequently,
\beq
\lb{N19}
\begin{split}
&\int_{\delta D}A^\nu(x^\mu,\phi,\partial_\mu\phi,s^\mu)n_\nu\; d\sigma =\int_{D}\partial_\nu A^\nu(x^\mu,\phi,\partial_\mu\phi,s^\mu)\; d^dx\\
&\;\;\;\;\;\;\;\;\;\;=\int_{D} \left[\frac{\partial \mathcal{L}}{\partial x^\nu}\xi^\nu+\frac{\partial \mathcal{L}}{\partial \phi}\eta+\frac{\partial \mathcal{L}}{\partial \left( \partial_\nu \phi\right)}\left(\frac{d\eta}{dx^\nu}-\partial_\mu \phi\frac{d\xi^\mu}{dx^\nu}\right)+\mathcal{L}\frac{d\xi^\nu}{dx^\nu}\right]e^{-f(x^\alpha)} \; d^dx =0.
\end{split}
\eeq
Let us now obtain the generalized Euler-Lagrange equation \eqref{GHEL} under the integral \eqref{N19}. By considering first the term involving $\frac{d\eta}{dx^\nu}$,  equation \eqref{N19} can be written as
\beq
\lb{N20}
\begin{split}
&\int_{\delta D}A^\nu(x^\mu,\phi,\partial_\mu\phi,s^\mu)n_\nu\; d\sigma =\\
&\;\;\;\;\;\;=\int_{D} \eta\left[\frac{\partial \mathcal{L}}{\partial \phi}-\frac{d}{dx^\nu}\frac{\partial \mathcal{L}}{\partial \left( \partial_\nu \phi\right)}+\gamma_\nu\frac{\partial \mathcal{L}}{\partial \left( \partial_\nu \phi\right)} \right]e^{-f(x^\alpha)} \; d^dx +\int_{D} \frac{d}{dx^\nu}\left[\frac{\partial \mathcal{L}}{\partial \left( \partial_\nu \phi\right)}\eta e^{-f(x^\alpha)}\right]\; d^dx+\\
&\;\;\;\;\;\;+\int_{D} \left[\frac{\partial \mathcal{L}}{\partial x^\nu}\xi^\nu-\frac{\partial \mathcal{L}}{\partial \left( \partial_\nu \phi\right)}\partial_\mu \phi\frac{d\xi^\mu}{dx^\nu}+\mathcal{L}\frac{d\xi^\nu}{dx^\nu}\right]e^{-f(x^\alpha)} \; d^dx =0,
\end{split}
\eeq
which on the solution of the generalized Euler-Lagrange equations \eqref{GHEL} becomes
\beq
\lb{N21}
\int_{D} \left[\frac{d}{dx^\nu}\left(\frac{\partial \mathcal{L}}{\partial \left( \partial_\nu \phi\right)}\eta e^{-f(x^ \alpha)}\right)+\left(\frac{\partial \mathcal{L}}{\partial x^\nu}\xi^\nu-\frac{\partial \mathcal{L}}{\partial \left( \partial_\nu \phi\right)}\partial_\mu \phi\frac{d\xi^\mu}{dx^\nu}+\mathcal{L}\frac{d\xi^\nu}{dx^\nu}\right)e^{-f(x^\alpha)}\right] \; d^dx =0.
\eeq
Finally, by considering the terms involving $\frac{d\xi^\mu}{dx^\nu}$ and $\frac{d\xi^\nu}{dx^\nu}$ in \eqref{N21} we get
\beq
\lb{N22}
\begin{split}
&\int_{D} \frac{d}{dx^\nu}\left[\left(\frac{\partial \mathcal{L}}{\partial \left( \partial_\nu \phi\right)}\eta + \mathcal{L}\xi^\nu -\frac{\partial \mathcal{L}}{\partial \left( \partial_\nu \phi\right)}\partial_\mu \phi \xi^\mu \right)e^{-f(x^\alpha)}\right] \; d^dx -\\
&\int_{D} \left[\xi^\nu\partial_\nu \phi\left( \frac{\partial \mathcal{L}}{\partial \phi}-\frac{d}{dx^\mu}\frac{\partial \mathcal{L}}{\partial \left( \partial_\mu \phi\right)}+\gamma_\mu \frac{\partial \mathcal{L}}{\partial \left( \partial_\mu \phi\right)}\right) + \gamma_\mu \partial_\nu s^\mu\xi^\nu - \mathcal{L} \gamma_\nu \xi^\nu \right]e^{-f(x^\alpha)} \; d^dx =0,
\end{split}
\eeq
which on the solution of the generalized Euler-Lagrange equations \eqref{GHEL}, and by using \eqref{H7}, reduces to
\beq
\lb{N23}
\begin{split}
&\int_{D}\bigg\{ \frac{d}{dx^\nu} \left[\left(\frac{\partial \mathcal{L}}{\partial \left( \partial_\nu \phi\right)}\eta + \mathcal{L}\xi^\nu -\frac{\partial \mathcal{L}}{\partial \left( \partial_\nu \phi\right)}\partial_\mu \phi \xi^\mu \right)e^{-f(x^\alpha)}\right]\\
&\qquad\qquad\qquad\qquad\qquad\qquad\qquad+e^{-f(x^\alpha)}\xi^\nu\left( \gamma_\nu \partial_\mu s^\mu - \gamma_\mu \partial_\nu s^\mu\right) \bigg\}\; d^dx=0.
\end{split}
\eeq
Since \eqref{N23} should be satisfied for any subdomain $D$ of $\Omega$, we obtain \eqref{GNT}.
\end{proof}

We can now formulate the Generalized Noether Theorem for Herglotz variational principle. Since the Action $S(\delta \Omega)$ is gauge invariant, the action-density field $s^ \mu$ is not uniquely defined by the Fundamental Problem given by \eqref{H7}. This is not a problem in determining the equation of motion \eqref{GHEL} that completely fix the physical field $\phi$, in a similar way that Maxwell's equations completely fix the physical electromagnetic field but do not fix the vector and scalar potentials. However, the specific choices we make for the action-density $s^\mu$ plays an important role when analyzing the symmetries from invariance transformations, since $s^\mu$ arouses in the identity \eqref{GNT} that holds when the Action is invariant under \eqref{N01}. In this context, in order to obtain conserved quantities from \eqref{GNT} with physically meaningful content, it is reasonable to choose a gauge where the identity \eqref{GNT} becomes a total derivative. Thus we have for the generalization of the Noether Theorem:
\begin{theorem}[Generalized Noether Theorem] 
\label{GHNT}
Let \rf{N01} be a symmetry group of the functional $S(\delta \Omega)$ defined in \rf{H7}, and let us assume the canonical gauge \eqref{CG}. Then the following quantity
\beq
\lb{HNT}
\left[\frac{\partial \mathcal{L}}{\partial \left( \partial_\nu \phi\right)}\eta +\mathcal{L}\xi^\nu -\frac{\partial \mathcal{L}}{\partial \left( \partial_\nu \phi\right)}\partial_\mu \phi\; \xi^\mu \right]e^{-f(x^\alpha)}
\eeq
is conserved (constant of motion) on solutions of the generalized Euler-Lagrange equation \rf{GHEL}.
\end{theorem}
\begin{proof}
The proof follows directly from \eqref{GNT} by inserting \eqref{CG}.
\end{proof}

\begin{remark} Note that the canonical gauge \eqref{CG} is not the only possibility to reduce the identity \eqref{GNT} into a total derivative. For example, the gauge $(\gamma_\nu\partial_\mu s^\mu-\gamma_\mu\partial_\nu s^\mu)\xi^\nu=e^{f(x^\alpha)}\partial_\nu F^\nu(x^\alpha)$, where $F^\nu(x^\alpha)$ is any vector field, also reduces \eqref{GNT} into a total derivative. However, there is no physical motivation to introduces a nonnull field $F^\nu(x^\alpha)$ in the problem.
\end{remark}

\begin{remark}
Our Generalized Noether Theorem \rf{HNT} generalizes for fields the Noether's like theorems for the classical Herglotz problem \cite{GGB,SMT}
\end{remark}

\begin{remark}
It is easy to see that for Lagrangian functions independent on $s^{\mu}$, the Generalized Noether Theorem \rf{HNT} reduces to the usual one,
\beq
\lb{NT}
\frac{\partial \mathcal{L}}{\partial \left( \partial_\nu \phi\right)}\eta +\mathcal{L}\xi^\nu -\frac{\partial \mathcal{L}}{\partial \left( \partial_\nu \phi\right)}\partial_\mu \phi\; \xi^\mu=\mbox{constant},
\eeq
since, in this case, $\gamma_\mu =0$ implies $f(x^\alpha)=\mbox{constant}$.
\end{remark}

Finally, it is straightforward to generalize Theorem \ref{GHNT} to the case with several fields $\phi^i(x^{\mu})=\phi^i(x^1,x^2,\cdots,x^d)$ ($i=1,...,N$). In this case we have the Action $S(\delta\Omega)$ in \eqref{H7B}. By defining a one-parameter group of invertible transformations
\beq
\label{N01B}
\left\{
\begin{array}{l}
\tilde{x}^\mu=\varphi^\mu(x^\nu,\phi^j;\epsilon),\\
\tilde{\phi^i}=\psi^i(x^\nu,\phi^j;\epsilon),
\end{array}
\right.
\eeq
where $\varphi^\mu(x^\nu,\phi^j;0)=x^\mu$ and $\psi^i(x^\nu,\phi^j;0)=\phi^i$, and the transformed action-density $\tilde{s}^\mu$ of $s^\mu$ is solution of $
\tilde{\partial}_\nu \tilde{s}^\nu=\mathcal{L}\left(\tilde{x}^\mu,\tilde{\phi}^j,\tilde{\partial}_\mu\tilde{\phi}^j,\tilde{s}^\mu\right)$ ($\tilde{x}^\mu\in \tilde{\Omega}$), we get
\begin{theorem}[Generalized Noether Theorem for several fields] 
\label{GHNTB}
Let \rf{N01B} be a symmetry group of the functional $S(\delta \Omega)$ defined in \rf{H7B}, and let us assume the canonical gauge \eqref{CG}. Then the following quantity
\beq
\lb{HNTB}
\left[\frac{\partial \mathcal{L}}{\partial \left( \partial_\nu \phi^i\right)}\eta^i +\mathcal{L}\xi^\nu -\frac{\partial \mathcal{L}}{\partial \left( \partial_\nu \phi^i\right)}\partial_\mu \phi^i\; \xi^\mu \right]e^{-f(x^\alpha)},
\eeq
where $\xi^\mu=\frac{d\varphi^\mu}{d\epsilon}\vert_{\epsilon=0}$ and $\eta^i=\frac{d\psi^i}{d\epsilon}\vert_{\epsilon=0}$, is conserved (constant of motion) on the solution of the generalized Euler-Lagrange equation \rf{GHEL}.
\end{theorem}
\begin{proof}
The proof follows similarly the ones of Theorem \ref{TNHG} and Theorem \ref{GHNT}, by considering the Definition \ref{ID}.
\end{proof}


\section{Examples}\lb{sec:EX}

The conserved quantities in any dynamical system play a major role in the analysis of the system. They enable us to solve some problems without a more detailed knowledge of the dynamics, for example, as found in any undergraduate textbook in physics, we can solve easily several mechanical problems by making use of the energy and momentum conservations without the necessity of solving the dynamical equation given by Newton's second law of motion. In general, the continuous symmetries and its related conserved quantities give us first integrals for dynamical systems that can be used to simplify the problem. In this context, our generalized Noether Theorems \ref{HNT} and \ref{HNTB} provide us a fairly automatic procedure to find conserved quantities for dissipative systems.  

In this section, we consider two examples in order to illustrate the potential of application of our Action Principle in Definition \ref{AP}, and its related Noether Theorems \ref{HNT} and \ref{HNTB}, to investigate dissipative systems. In the first, we investigate the conserved quantities related to the symmetries under space and time transformations for a vibrating string under viscous forces. The second example illustrates a conservation law related to internal (global) symmetry in a dissipative complex scalar field system.

\subsection{Spacetime transformations symmetries: a vibrating string under viscous forces}

In order to illustrate the potential of application of our Action Principle \ref{AP} and generalized Noether theorems \eqref{GHNT} and \eqref{GHNTB} to investigate dissipative systems, we consider a vibrating string under viscous forces (like the frictional reaction of the air through which the string moves, among others). This is the simplest continuous mechanical system that we can include dissipative forces. We can also extend this method to bars, membranes, etc.
Let us consider a two-dimensional space-time ($d=2$), with $x_1=t$ ($t\in [t_a,t_b]$), and $x_2=x$, ($x\in [a,b]$). The Lagrangian function for a vibrating string under viscous forces can be given by \cite{MJJG2}
\beq
\lb{E1a}
\mathcal{L}=\frac{\mu}{2}\left(\partial_t \phi\right)^2-\frac{T}{2}\left(\partial_x \phi\right)^2-\frac{\gamma}{\mu}s^1
\eeq
where $\mu$ is the mass density, $T$ is the string tension, $\phi$ is the string transverse displacement from equilibrium, $\gamma$ is the viscous coefficient of the medium, and we choose $\gamma_\mu=(-\frac{\gamma}{\mu},0)$. The last term in \eqref{E1a} can be interpreted as potential energy for the dissipative force \cite{MJJG2}. The first and second terms in \eqref{E1a} are the kinetic energy and the elastic potential, respectively. From the Lagrangian function \eqref{E1a}, it is easy to see that our Action Principle gives the correct equation of motion for a string under the presence of a viscous force proportional to the first order derivative $\partial_t \phi$. By inserting \eqref{E1a} into the generalized Euler-Lagrange equation \eqref{GHEL} we get
\beq
\lb{E1b}
\mu \partial_{tt} \phi-T\partial_{xx}\phi+\gamma \partial_t \phi=0.
\eeq

Since we have a frictional force, the total energy of the system is not conserved. From our generalized Noether's Theorem \ref{GHNT} the conserved quantity under time and space translations ($\xi^\mu\neq 0$ and $\eta=0$) is 
\beq
\lb{E1c}                      
T^\nu_\mu e^{\frac{\gamma}{\mu} t}=\left(\frac{\partial \mathcal{L}}{\partial \left( \partial_\nu \phi\right)}\partial_\mu \phi- \delta_\mu^\nu \mathcal{L} \right)e^{\frac{\gamma}{\mu} t},
\eeq
where $T_\mu^\nu$ is the well know stress-energy tensor for a scalar field $\phi$, $\delta_\mu^\nu$ is the Kronecker delta function, and we set $f(x^\alpha)=-\frac{\gamma}{\mu} t$ since $\partial_\nu f =\partial_{s^\nu}\mathcal{L}=\gamma_\nu=(-\frac{\gamma}{\mu},0) $ is a constant vector. Since \eqref{E1c} is conserved, we have
\beq
\lb{E1d}
\partial_\nu \left(T^\nu_\mu e^{\frac{\gamma}{\mu} t}\right)=0.
\eeq
Then, the two conserved quantities are given by 
\beq
\lb{E1e}
Ee^{\frac{\gamma}{\mu} t}\;\;\;\;\; \mbox{and}\;\;\;\;\; Pe^{\frac{\gamma}{\mu} t},
\eeq
where 
\beq
\lb{E1f}
E(t)=\int_a^b T^1_1dx=\int_a^b\left(\frac{\mu}{2}\left(\partial_t \phi\right)^2+\frac{T}{2}\left(\partial_x \phi \right)^2+\frac{\gamma}{\mu}s^1\right)dx 
\eeq
is the total energy (sum of kinetic and potential energies), and
\beq
\lb{E1g}
 P(t)=\int_a^b T^0_1dx=\mu\int_a^b \partial_t \phi\; \partial_x \phi\; dx
\eeq
is the total momentum of the system. Then, from the conserved quantities \eqref{E1e} we can conclude that the value of both energy and momentum decreases exponentially in time. In particular, we have for the energy
\beq
\lb{E1h}
E(t)=E_0e^{-\frac{\gamma}{\mu} t}=e^{-\frac{\gamma}{\mu}}\int_a^b\left.\left(\frac{\mu}{2}\left(\partial_t \phi\right)^2+\frac{T}{2}\left(\partial_x \phi \right)^2\right)\right|_{t=0}dx,
\eeq
where $E_0$ is the initial value of the mechanical energy, and, since the Lagrangian function is defined less than a constant (actually, it is defined less than a total derivative) as in traditional calculus of variation, we set $s^1|_{t=0}=0$ without loss of generality.

\subsection{Time transformations symmetry: a two-degree-of-freedom nonlinear dissipative oscillator}

In the present example, we consider a spherical pendulum under frictional forces. A spherical pendulum is a simple pendulum consisting of a particle of mass $m$ suspended from a fixed point $O$ by a rigid rod of length $l$ and negligible mass. The pendulum is free to swing to the entire solid angle about the point $O$. Consequently, the particle of mass $m$ moves on a spherical surface of radius $l$ in the gravitational field. A Lagrangian for this system, in a spherical coordinate system, is given by
\beq
\lb{E3a}
\mathcal{L}=\frac{ml^2}{2}\left(\dot{\theta}^2+\sin^2(\theta){\dot{\phi}}^2\right)+mgl\cos(\theta)-\frac{\gamma}{ml} s,
\eeq
where on this case $x_1=t$ ($d=1$), $\gamma_1=-\frac{\gamma}{ml}$, $\phi^1=\theta$ and $\phi^2=\phi$ ($N=2$). The variable $\theta$ is the polar angle (the angle between the vertical line and the rigid rod), and $\phi$ is the azimuthal angle (the rotation angle about the vertical line). As in the previous example, the last term in \eqref{E3a} can be interpreted as the potential energy of the dissipative forces acting in the particle. The remaining terms are the kinetic energy minus the potential energy of the conservative gravitational force for the spherical pendulum \cite{Symon}. From our generalized Euler-Lagrange equation \eqref{GHEL2} we obtain the following equations of motion for the pendulum
\beq
\lb{E3b}
ml^2\ddot{\theta}+\gamma l\dot{\theta}-ml^2\sin(\theta)\cos(\theta)\dot{\phi}^2+mgl\sin(\theta)=0
\eeq
and
\beq
\lb{E3c}
\frac{d}{dt}\left(ml^2\sin^2(\theta)\dot{\phi}\right)+\gamma l\sin^2(\theta)\dot{\phi}=0.
\eeq
For $\gamma=0$, \eqref{E3b} and \eqref{E3c} reduce to the well know equations of motion for the classical conservative spherical pendulum \cite{Symon}. On the other hand, as in the simple example of a particle under frictional forces discussed in Section \ref{sec:APb}, the effective potential energy $\frac{\gamma}{ml}s$ for the dissipative forces in the Lagragian \eqref{E3a} introduces frictional forces proportional to the velocity in the equations of motion. In \eqref{E3b} we have a frictional force proportional to the polar velocity $v_\theta=l\dot{\theta}$ and in \eqref{E3c} a force proportional to the azimuthal velocity $v_\phi=l\sin(\theta)\dot{\phi}$. As a consequence, the total energy of the system, as well as the azimuthal angular momentum $p_\phi=\frac{\partial L}{\partial \dot{\phi}}=ml^2\sin^2(\theta)\dot{\phi}$, is not conserved. It is important to notice that for $\gamma=0$ the equation \eqref{E3c} reduces to the conservation of the azimuthal angular momentum $\frac{dp_\phi}{dt}=0$, as we should expect for a dissipationless system. From our generalized Noether's Theorem \ref{GHNTB} the conserved quantity under time translations ($\xi^1\neq 0$ and $\eta=0$) is
\beq
\lb{E3d}
H(t)e^{\frac{\gamma}{ml}t}=\left(\frac{ml^2}{2}\left(\dot{\theta}^2+\sin^2(\theta){\dot{\phi}}^2\right)-mgl\cos(\theta)+\frac{\gamma}{ml} s\right)e^{\frac{\gamma}{ml}t},
\eeq
from where it is evident that the Hamiltonian $H(t)=E(t)$, corresponding to the total energy $E$ of the system (kinetic energy plus the potential energy of both conservative and dissipative forces), is not conserved if $\gamma\neq 0$.  Actually, as in the previous example we have for the energy 
\beq
\lb{E3e}
E(t)=E_0e^{-\frac{\gamma}{ml}t}=\left.\left(\frac{ml^2}{2}\left(\dot{\theta}^2+\sin^2(\theta){\dot{\phi}}^2\right)-mgl\cos(\theta)\right)\right|_{t=0}e^{\frac{\gamma}{ml}t},
\eeq
where $E_0$ is a constant. Consequently, the total energy decreases exponentially with time. Finally, the conserved quantity \eqref{E3d} obtained from our generalized Noether's Theorem \ref{GHNTB} provide us a relation \eqref{E3e} that can be used to eliminate $\dot{\phi}$ in the equations of motion \eqref{E3b} and \eqref{E3c} facilitating the solution of this nonlinear problem.

\subsection{Internal symmetry: a dissipative complex scalar field}

A complex scalar field is the simplest problem displaying internal symmetry. It appears in the description of quantum systems, where the complex field $\phi$ describes the wave function of a physical scalar field related to bosonic particles and its anti-particles. It arises in the description of several quantum systems, for example, in the description of the collective excitation (phonon) in periodic elastic arrangement of atoms/molecules in condensed matter (solids and some liquids). In order to consider the simplest dissipative (open) quantum system displaying internal symmetry, let us consider the following Lagrangian function
\beq
\lb{E2a}
\mathcal{L}=\partial_\mu \phi \partial^\mu \phi^*-m^2\phi \phi^*-\gamma_\mu s^\mu,
\eeq
where $\phi^*$ is the complex conjugate of $\phi$, $m$ is the mass density of the field $\phi$, and the last term in \eqref{E2a} can be interpreted as a potential energy of a dissipative interaction. From the generalized Euler-Lagrange equation \eqref{GHEL2} we obtain, by considering $\phi^1=\phi$ and $\phi^2=\phi^*$, the following equation of motion:
\beq
\lb{E2b}
\partial_\mu \partial^\mu \phi+m\phi-\gamma_\mu \partial^\mu \phi=0,
\eeq
that for $\gamma_\mu=0$, and $\partial_\mu \partial^\mu=\frac{1}{c^2}\frac{\partial^2}{\partial t^2}-\nabla^2$, reduces to the well know Klein-Gordon equation. 

The lagrangian \eqref{E2a} has a continuous symmetry related to phase changes of $\phi$, since the transformation $\tilde{x}^\mu=x^\mu$ and $\tilde{\phi}=e^{i\varepsilon}\phi$ ($\xi^\mu=0$, $\eta^1=i\phi$, and $\eta^2=-i\phi^*$) leaves \eqref{E2a} invariant. Thus, from the generalized Noether theorem \eqref{HNTB} we obtain the following associated conserved current
\beq
\lb{E2c}
j^\mu=i\left(\phi\partial^\mu \phi^*-\phi^*\partial^\mu \phi\right)e^{\gamma_\alpha x^\alpha}.
\eeq
Since currents of the form $i\left(\phi\partial^\mu \phi^*-\phi^*\partial^\mu \phi\right)$ have the interpretation of electric charge (or particle number), from \eqref{E2c} we see that, as we should expect in an dissipative (open) problem, the charge of a system defined by the Lagrangian \eqref{E2a} decreases exponentially in time when $\gamma_\mu=(\gamma,0,0,0)$.


\section{Conclusions}\lb{sec:C}

In the present work, we formulate a generalization of the Noether Theorem for the Action Principle with action-dependent Lagrangian functions introduced in \cite{MJJG,MJJG2}. When the dependence on the action is removed, both the Action Principle and the generalized Noether Theorem reduces to the traditional ones. Noether's theorem is one of the most important theorems for physics in the 20th century. It is well known that all conservation laws in physics, \textrm{e.g.}, conservation of energy or conservation of momentum, are directly related to the invariance of the action under a family of transformations. However, the classical Noether theorem cannot yield information about constants of motion for non-conservative systems since it is not possible to formulate physically meaningful Lagrangians for this kind of systems in the classical calculus of variation. On the other hand, our recent Action Principle with action-dependent Lagrangian functions \cite{MJJG,MJJG2} enables us to construct meaningful Lagrangian functions, which provide physically consistent expressions for the momentum and the Hamiltonian of the system, for a huge variety of non-conservative systems (classical and quantum). Consequently, the generalized Noether Theorem we formulate in the present work enables us to investigate conservation laws for non-conservative systems. 
In order to illustrate the potential of application of our Action Principle and its related Noether Theorem, we consider three examples of dissipative systems. In the first, we investigate the conserved quantities related to spacetime transformations symmetries for a viscous vibrating string. In the second example, we studied the conservation law related to time transformations symmetry in a two-degree-of-freedom nonlinear dissipative oscillator. Finally, in the last example, we analyze the conservation law related to internal (global) symmetry of a dissipative complex scalar field.


\section*{acknowledgments}
This work was partially supported by CNPq and CAPES (Brazilian research funding agencies). 


\section*{Conflict of Interest}

Conflict of Interest: The authors declare that they have no conflict of interest.


\end{document}